\documentclass[letterpaper,reqno,11pt]{article}
\usepackage[margin=1.0in]{geometry}
\usepackage{color,latexsym,amsmath,amssymb}
\usepackage{fancyhdr}
\usepackage{amsthm}
\usepackage[linesnumbered,lined,boxed,commentsnumbered,ruled]{algorithm2e}
\usepackage{dsfont}
\usepackage{graphicx}
\usepackage[hidelinks=true]{hyperref}
\usepackage{setspace}
\usepackage{lmodern}
\usepackage{thmtools}
\usepackage{thm-restate}
\usepackage[numbers]{natbib}
\usepackage{subcaption}
\usepackage[inline]{enumitem}
\usepackage{nicematrix}
\usepackage{cleveref}
\usepackage{booktabs}
\usepackage{multirow}
\usepackage{makecell}

\usepackage{tikz}
\usetikzlibrary{bbox}
\usetikzlibrary{fit}
\usetikzlibrary{arrows.meta}

\usepackage{pgfplots}
\pgfplotsset{compat=1.18}
\usepgfplotslibrary{fillbetween}

\usepackage{listings}
\crefname{subsection}{subsection}{subsections}
\Crefname{subsection}{Subsection}{Subsections}

\allowdisplaybreaks

\newcommand{\RR}{\mathbb{R}}

\newcommand{\ZZ}{\mathbb{Z}}
\newcommand{\QQ}{\mathbb{Q}}
\newcommand{\NN}{\mathbb{N}}
\newcommand{\T}{\mathsf{T}}

\newcommand{\maximize}{\textrm{maximize}}
\newcommand{\subjectto}{\textrm{subject to}}

\newcommand{\argmax}{\mathrm{arg\,max}}
\newcommand{\OPT}{\mathsf{OPT}}
\newcommand{\LPOPT}{\mathsf{LPOPT}}

\DeclareMathOperator{\LP}{\mathsf{LP}}

\DeclareMathOperator{\supp}{\mathsf{supp}}

\DeclareMathOperator{\prev}{\mathsf{prev}}
\DeclareMathOperator{\diag}{\mathsf{diag}}

\title{Bicriteria Approximation Algorithms for Demand Matching}
\author{Yuchong Pan \\ \small MIT \\ \small \tt yuchong@mit.edu \and Michel X.\ Goemans \\ \small MIT \\ \small \tt goemans@math.mit.edu}
\date{}
\newtheorem{theorem}{Theorem}[section]
\newtheorem{lemma}[theorem]{Lemma}

\newtheorem{corollary}[theorem]{Corollary}

\theoremstyle{definition} 

\begin{document}

\maketitle

\begin{abstract}
    The demand matching problem generalizes both the knapsack problem and the $b$-matching problem. In this problem, each edge of a graph has a demand and a weight, each vertex has a capacity, and the goal is to find a maximum weight subset of edges whose total incident demand at every vertex does not exceed its capacity. We study $(\alpha, \beta)$-bicriteria approximation algorithms, which return a solution of weight at least $1/\alpha$ times the optimum while allowing an additive capacity violation of at most $\beta$ times the maximum edge demand.

We give an iterative relaxation algorithm for the demand matching problem that exploits a structural characterization of strictly fractional extreme points of the natural LP relaxation, which reduces the residual rounding problem to odd-cycle instances. Combined with a better-of-two rounding strategy, this yields $(7/6, 1)$- and $(1, 1)$-bicriteria approximation algorithms for general and bipartite graphs, respectively. We further generalize this approach to obtain a parametric family of algorithms, including a $(1, 4/3)$-bicriteria approximation. Separately, for the more general $k$-hypergraph demand matching problem, we give a greedy, combinatorial $(k, 1)$-bicriteria approximation algorithm.

We complement these algorithmic results with matching lower bounds relative to the natural LP relaxation for $\beta = 0$ and all $\beta \geq 1$, completely characterizing the trade-off between weight approximation and additive capacity violation in this range.
\end{abstract}

\section{Introduction}

We consider the \textsc{Demand Matching} problem introduced by \citet{shepherd2007demand}. Let $G = (V, E)$ be a graph. Let $b : V \to \QQ_{\geq 0}$ denote \emph{vertex capacities}. Let $d : E \to \QQ_{> 0}$ and $w : E \to \QQ_{\geq 0}$ denote \emph{edge demands} and \emph{edge weights}, respectively. We say that a subset $M \subseteq E$ is a \emph{demand matching} if $d(M \cap \delta(v)) \leq b(v)$ for all $v \in V$, where $\delta(v)$ denotes the set of edges of $G$ incident to $v$.\footnote{Given a set $\Sigma$ and a function $\varphi : \Sigma \to \QQ$, we write $\varphi(S) := \sum_{s \in S} \varphi(s)$ for all $S \subseteq \Sigma$.} To emphasize the vertex capacities $b$ and the edge demands $d$, we might say that $M$ is \emph{$(b,d)$-feasible}. The objective of the \textsc{Demand Matching} problem is to find a demand matching $M$ such that $w(M)$ is maximized. A natural LP relaxation of the \textsc{Demand Matching} problem is the following:
\begin{alignat}{4}
    \maximize \qquad && \sum_{e \in E} x_e w(e) & \tag{DM-LP} \label{eq:lp} \\
    \subjectto \qquad && \sum_{e \in \delta(v)} x_e d(e) & \leq b(v) && \qquad \forall v \in V, \notag \\
    && x_e &\in [0, 1] && \qquad \forall e \in E. \notag
\end{alignat}
Throughout this paper, we assume that, for each $e = uv \in E$, we have $d(e) \leq \min\{ b(u), b(v) \}$. This \emph{no-bottleneck assumption} can be made without loss of generality without affecting the optimum of the original instance, although it may affect the optimum of the LP relaxation. We also assume, without loss of generality, that $G$ does not contain self-loops. Multiple edges are allowed between the same pair of vertices.

A well-known special case of the \textsc{Demand Matching} problem is the case where $d(e) = 1$ for all $e \in E$. This is the familiar \textsc{Maximum Weight $b$-Matching} problem, which admits several polynomial-time algorithms \cite{marsh1979matching,padberg1982odd,gabow1983efficient,anstee1987polynomial,gerards1995matching}. In contrast, when $G$ has only two vertices, the \textsc{Demand Matching} problem reduces to the \textsc{Knapsack} problem. Hence, the \textsc{Demand Matching} problem is $\mathsf{NP}$-hard. Moreover, \citet{shepherd2007demand} showed that even the cardinality case of the \textsc{Demand Matching} problem (i.e., $w(e) = 1$ for all $e \in E$) is $\mathsf{MAXSNP}$-hard; in other words, there exists $\varepsilon > 0$ such that the \textsc{Demand Matching} problem admits no $(1+\varepsilon)$-approximation algorithm unless $\mathsf{P} = \mathsf{NP}$.

The \textsc{Demand Matching} problem has been extensively studied, and several approximation algorithms are known. For $\alpha \geq 1$, we say that an algorithm is an \emph{$\alpha$-approximation} for the \textsc{Demand Matching} problem if, given an instance $(G = (V, E), b, d, w)$ with optimum $\OPT$, it outputs in polynomial time a $(b, d)$-feasible demand matching $M \subseteq E$ such that $w(M) \geq \OPT/\alpha$. We say that an approximation algorithm for the \textsc{Demand Matching} problem is \emph{with respect to \eqref{eq:lp}} if $\OPT$ in the above definition can be replaced by the optimum of \eqref{eq:lp}. \citet{shepherd2007demand} gave a $3.264$-approximation algorithm for general graphs and a $2.764$-approximation algorithm for bipartite graphs, both with respect to \eqref{eq:lp}, and also showed that the \textsc{Demand Matching} problem is $\mathsf{APX}$-hard even in bipartite graphs. \citet{parekh2011iterative} improved the upper bound on the integrality gap of \eqref{eq:lp} to $3$, which matches the lower bound by \citet{shepherd2007demand}. For bipartite graphs, \citet{singh2012nearly} improved the upper bound on the integrality gap to $2.709$, with a lower bound of $2.699$.

\subsection{Main Contribution}

The motivation of this work is to allow a (small) additive violation on each vertex capacity, in order to reduce the approximation ratio on the weight. In other words, we give \emph{resource augmentation} results, where we compare a solution that may use \emph{augmented} capacities against the optimum under the original capacities. For $\alpha \geq 1$ and $\beta \geq 0$, we say that an algorithm is an \emph{$(\alpha, \beta)$-bicriteria approximation} for the \textsc{Demand Matching} problem if, given an instance $(G = (V, E), b, d, w)$ with optimum $\OPT$, it outputs in polynomial time a $(b + \beta d_{\max}, d)$-feasible demand matching $M \subseteq E$ such that $w(M) \geq \OPT/\alpha$, where $d_{\max} := \max_{e \in E} d(e)$.\footnote{Given a set $\Sigma$, a function $\varphi : \Sigma \to \QQ$, and a constant $\gamma \in \QQ$, we use $\varphi + \gamma$ to denote the function defined by $s \mapsto \varphi(s) + \gamma$ for all $s \in \Sigma$.} Resource augmentation results have been established for some other combinatorial optimization problems including the \textsc{Minimum Cost Bounded Degree Spanning Tree} problem \cite{konemann2000matter,konemann2003primal,ravi2006delegate,goemans2006minimum, chaudhuri2009would,singh2015approximating}, the \textsc{Single-Source Unsplittable Flow} problem \cite{dinitz1999single,kolliopoulos2001approximation,skutella2002approximating,martens2007convex,morell2022single,traub2024single,aleman2025unsplittable,majthoub2025integer}, and the \textsc{$k$-Edge-Connected Spanning Subgraph} problem \cite{hershkowitz2024ghost,nutov2025bicriteria,kumar2025almost}. These works have introduced a variety of techniques for obtaining resource augmentation results.

To the best of our knowledge, very few bicriteria approximation results are known for the \textsc{Demand Matching} problem. One example is the work of \citet{ahmadian2017further}, who studied a generalization of the \textsc{Demand Matching} problem (with asymmetric demands at the two endpoints for each edge and with a matroid constraint) and used the iterative relaxation technique to obtain a $(1, 2)$-bicriteria approximation algorithm.

The main contribution of this paper is a collection of bicriteria approximation algorithms for the \textsc{Demand Matching} problem. \Cref{tab:summary} provides a partial summary of our algorithmic results, together with the previously best known LP-relative approximation results for comparison. The remainder of this section states our main theorems and highlights the key technical innovations underlying them.

\begin{table}[ht]
    \crefname{theorem}{thm.}{thms.}
    \Crefname{theorem}{Thm.}{Thms.}
    \centering
    \begin{tabular}{c|c|c|c}
        Problem & Best known LP-rel.\ approx.\ & This paper & Technique \\
        \hline \hline
        \multirow{2}{*}{\textsc{DM}} & \multirow{2}{*}{$3$ \cite{parekh2011iterative} (LB: $3$ \cite{shepherd2007demand})} & $(7/6, 1)$ (\Cref{thm:better2}) & \multirow{3}{*}{\makecell{iterative relaxation\\+ better-of-two}} \\
        \cline{3-3}
        & & $(1,4/3)$ (\Cref{thm:parametric-nba}) & \\
        \cline{3-3}
        \cline{1-3}
        bipartite \textsc{DM} & $2.709$ \cite{singh2012nearly} (LB: $2.699$ \cite{singh2012nearly}) & $(1, 1)$ (\Cref{thm:better2}) & \\
        \hline
        \multirow{2}{*}{\textsc{$k$-HDM}} & $2k$ \cite{parekh2011iterative,parekh2014generalized} & \multirow{2}{*}{$(k, 1)$ (\Cref{thm:main})} & \multirow{2}{*}{greedy} \\
        & (LB: $2(k - 1 + 1/k)$ \cite{parekh2011iterative}) & &
    \end{tabular}
    \caption{A partial summary of main algorithmic results in this paper. We use a number $\alpha$ to denote an $\alpha$-approximation algorithm, and a pair $(\alpha, \beta)$ to denote an $(\alpha, \beta)$-bicriteria approximation algorithm. Here, \textsc{$k$-HDM} denotes the \textsc{$k$-Hypergraph Demand Matching} problem, \textsc{DM} denotes the \textsc{Demand Matching} problem, ``bipartite'' denotes the special case where the graph is bipartite, and ``LB'' denotes the best known lower bound with respect to \eqref{eq:lp} in the first two rows and with respect to \eqref{eq:hdm-lp} in the third row.}
    \label{tab:summary}
\end{table}

\subsubsection{$(7/6, 1)$-Bicriteria Approximation Algorithm with Iterative Relaxation}

Our first result is a $(7/6, 1)$-bicriteria approximation algorithm with respect to \eqref{eq:lp} for the \textsc{Demand Matching} problem.

\begin{theorem} \label{thm:better2}
    For the \textsc{Demand Matching} problem, there is a $(7/6, 1)$-bicriteria approximation algorithm with respect to \eqref{eq:lp}. For an instance $(G, b, d, w)$ with $G$ bipartite, this algorithm gives a $(1, 1)$-bicriteria approximation.
\end{theorem}

The algorithm in \Cref{thm:better2} is based on the iterative relaxation framework. The iterative rounding technique for designing approximation algorithms was pioneered by \citet{jain2001factor} and later extended by \citet{singh2015approximating} into what is now known as iterative relaxation; see \citet{lau2011iterative} for a comprehensive overview. At a high level, an iterative relaxation algorithm repeatedly solves an appropriately chosen LP relaxation of a combinatorial optimization problem and gradually constructs a desired solution. In each iteration, progress is made by either eliminating an edge with LP value $0$, or fixing an edge with integral LP value, or dropping a nearly satisfied constraint, thereby simplifying the LP while approximately preserving feasibility. The framework of iterative relaxation has been successfully applied to a wide range of bicriteria approximation algorithms.

Our bicriteria approximation algorithms in \Cref{thm:better2} make use of a structural lemma for a \emph{strictly fractional} extreme point of \eqref{eq:lp} and its residual versions, i.e., an extreme point $x^\star$ such that $0 < x_e^\star < 1$ for every edge $e$ in the residual graph. Earlier, \citet{ahmadian2017further} proved a slightly weaker structural result for a strictly fractional extreme point $x^\star \in \RR^E$ of a residual LP, namely that $|\delta(v)| \leq x^\star(\delta(v)) + 2$ for \emph{some} nonrelaxed vertex $v$, and used this structure to design a $(1, 2)$-bicriteria approximation algorithm with the iterative relaxation framework. (Their result applies to a generalization of the \textsc{Demand Matching} problem with asymmetric demands at the two endpoints for each edge and an additional matroid constraint.)

In our setting, we obtain a more precise characterization of a strictly fractional extreme point of a residual LP arising in the iterative relaxation framework. In particular, after all integral edges have been fixed and all constraints at vertices of residual degree at most one have been dropped, the support graph of the remaining strictly fractional extreme point must be a vertex-disjoint union of odd cycles (with all other vertices isolated). We formally state and prove this characterization in \Cref{lem:extreme-point}. Through the iterative relaxation framework, this characterization reduces the residual rounding problem to a collection of odd-cycle instances. To prove \Cref{thm:better2}, we then apply the better of the following two rounding strategies:
\begin{itemize}
    \item[(S1)] Take a maximum weight demand matching on each odd cycle, with each vertex capacity being the residual capacity \emph{relaxed by $d_{\max}$}.
    \item[(S2)] Take all edges of the odd cycles, and discard all previously fixed edges.
\end{itemize}

The analysis of odd-cycle instances relies on two key technical ingredients. First, we observe that the special case of the \textsc{Demand Matching} problem on a cycle reduces to the well-studied \textsc{Maximum Weight $b$-Matching} problem under the no-bottleneck assumption. It is known that the latter problem can be solved in polynomial time. Second, and more importantly, we relate an optimal extreme point solution of \eqref{eq:lp} to an optimal extreme point solution of the fractional matching LP. This connection allows us to exploit structural properties of the fractional matching polytope to establish an inequality relating the following three quantities: the optimum of the demand matching LP, the maximum weight of a demand matching with relaxed capacities, and the total weight of all edges. The details are presented in \cref{sec:cycle}.

We remark that the same argument extends the $(7/6, 1)$-bicriteria approximation algorithm, with the same guarantee, to a generalization of the \textsc{Demand Matching} problem in which each edge $e = uv \in E$ is allowed to have different demand values $d^u(e)$ and $d^v(e)$ at its two endpoints, and the output $M \subseteq E$ must satisfy $d^v(M \cap \delta(v)) := \sum_{e \in M \cap \delta(v)} d^v(e) \leq b(v)$ for all $v \in V$. Indeed, in this generalized setting, the support graph of any strictly fractional extreme point of the residual LP is a vertex-disjoint union of cycles, which may be odd or even. This structural property suffices to obtain the $(7/6, 1)$-bicriteria approximation guarantee. Note, however, that the $(1, 1)$-bicriteria approximation algorithm for bipartite graphs does not extend to this more general setting. We omit the details in this paper.

\subsubsection{Parametric Bicriteria Approximation Algorithms and Matching Lower Bounds}

Two natural trade-off questions arise from \Cref{thm:better2}:
\begin{enumerate*}[label=(\roman*), itemsep=0pt]
    \item If an additive capacity violation of $\beta d_{\max}$ is allowed at each vertex for some $\beta \geq 0$, how well can the optimal weight be approximated?
    \item To achieve the optimal weight (or more generally, an approximation ratio $\alpha \geq 1$ on the weight), how much additive capacity violation is needed at each vertex?
\end{enumerate*}

Indeed, the algorithm in \Cref{thm:better2} can be generalized to a family of \emph{parametric} bicriteria approximation algorithms which characterize the trade-off between the weight approximation ratio and the additive violation of capacity constraints. This generalization relies on the same core ideas as the algorithm in \Cref{thm:better2}. In particular, it uses the iterative relaxation framework, and its analysis follows from a parametric generalization of the analysis of the \textsc{Demand Matching} problem on a cycle.

\begin{theorem} \label{thm:parametric-nba}
    For the \textsc{Demand Matching} problem, for $\beta \geq 1$, there is an $(\alpha^\mathrm{UB}(\beta), \beta)$-bicriteria approximation algorithm with respect to \eqref{eq:lp}, where
    \begin{equation} \label{eq:alpha-ub}
        \alpha^\mathrm{UB}(\beta) := \left\{
        \begin{array}{ll}
            (10 - 3\beta)/6 & \text{if $1 \leq \beta \leq \frac{4}{3}$}, \\
            1 & \text{otherwise}.
        \end{array}
    \right.
    \end{equation}
    In particular, there is a $(1, 4/3)$-bicriteria approximation algorithm for the \textsc{Demand Matching} problem.
\end{theorem}

We complement \Cref{thm:parametric-nba} with the following tightness result.

\begin{theorem} \label{thm:tight-nba}
    Let $\beta \geq 0$. For all $\varepsilon > 0$, there exists an instance $(G, b, d, w)$ of the \textsc{Demand Matching} problem satisfying the no-bottleneck assumption for which every $(b + \beta d_{\max}, d)$-feasible demand matching has weight at most $(1/\alpha^\mathrm{LB}(\beta) + \varepsilon)$ times the optimum of \eqref{eq:lp}, where
    $$ \alpha^\mathrm{LB}(\beta) := \left\{
        \begin{array}{ll}
            3(2 - \beta)/2 & \text{if $0 \leq \beta < 1$}, \\
            (10 - 3\beta)/6 & \text{if $1 \leq \beta \leq \frac{4}{3}$}, \\
            1 & \text{otherwise}.
        \end{array}
    \right. $$
    Consequently, for $\alpha < \alpha^\mathrm{LB}(\beta)$, there is no $(\alpha, \beta)$-bicriteria approximation algorithm with respect to \eqref{eq:lp} for the \textsc{Demand Matching} problem.
\end{theorem}

\Cref{thm:parametric-nba,thm:tight-nba}, together with the $3$-approximation algorithm of \citet{parekh2011iterative}, provide matching upper and lower bounds on the best possible weight approximation ratio $\alpha(\beta)$ as a function of $\beta$, for $\beta = 0$ and all $\beta \geq 1$. Moreover, the bounds in \Cref{thm:parametric-nba,thm:tight-nba} reveal a \emph{phase transition} at $\beta = 1$. For $0 \leq \beta < 1$, the lower bound is witnessed by triangle instances as shown in the proof of \Cref{thm:tight-nba}; as $\beta$ approaches $1$ from below, this lower bound approaches $3/2$. In contrast, once each vertex is allowed an additive violation of the maximum demand, the approximation ratio abruptly drops to $7/6$. This discontinuity reflects a structural change in the rounding problem. For $\beta < 1$, even on a triangle, one cannot take two incident edges that exceed the original capacity by a full demand unit. At $\beta = 1$, the strategy (S2), which takes all edges of the residual odd cycles, becomes available and can be combined with strategy (S1) through a better-of-two argument.

\subsubsection{Greedy $(k, 1)$-Bicriteria Approximation for $k$-Hypergraph Demand Matching}

Our last algorithmic result is a simple, combinatorial $(2, 1)$-bicriteria approximation algorithm with respect to \eqref{eq:lp} for the \textsc{Demand Matching} problem. Although the guarantee is measured against \eqref{eq:lp}, the algorithm is combinatorial; \eqref{eq:lp} is used only in the analysis of the weight guarantee.

More generally, we give a $(k, 1)$-bicriteria approximation algorithm for the \textsc{$k$-Hypergraph Demand Matching} problem, which we define below. In the \textsc{$k$-Hypergraph Demand Matching} problem with $k \geq 2$, we are given a hypergraph $H = (V, \mathcal E)$ such that each hyperedge contains at most $k$ vertices, vertex capacities $b : V \to \QQ_{\geq 0}$, hyperedge demands $d : \mathcal E \to \QQ_{> 0}$, and hyperedge weights $w : \mathcal E \to \QQ_{\geq 0}$, and the objective is to find a subset $\mathcal M \subseteq \mathcal E$ such that $d(\mathcal M \cap \delta(v)) \leq b(v)$ for all $v \in V$, and that $w(\mathcal M)$ is maximized; here $\delta(v)$ refers to the set of hyperedges containing $v$. When $k = 2$, this is the \textsc{Demand Matching} problem. A natural LP relaxation of the \textsc{$k$-Hypergraph Demand Matching} problem is the following:
\begin{alignat}{4}
    \maximize \qquad && \sum_{e \in \mathcal E} x_e w(e) & \tag{$k$-HDM-LP} \label{eq:hdm-lp} \\
    \subjectto \qquad && \sum_{e \in \delta(v)} x_e d(e) & \leq b(v) && \qquad \forall v \in V, \notag \\
    && x_e &\in [0, 1] && \qquad \forall e \in \mathcal E. \notag
\end{alignat}
The notions of $\alpha$-approximation algorithms and $(\alpha, \beta)$-bicriteria approximation algorithms for the \textsc{$k$-Hypergraph Demand Matching} problem are defined analogously to those for the \textsc{Demand Matching} problem, where we redefine $d_{\max} := \max_{e \in \mathcal E} d(e)$. We use the following natural extension of the \emph{no-bottleneck assumption} for hypergraphs: for every $e \in \mathcal E$ and every $v \in e$, we have $d(e) \leq b(v)$. \citet{parekh2011iterative} and \citet{parekh2014generalized} gave two $2k$-approximation algorithms for the \textsc{$k$-Hypergraph Demand Matching} problem with respect to \eqref{eq:hdm-lp}, while the integrality gap is at least $2(k - 1 + 1/k)$.

\begin{theorem} \label{thm:main}
    For the \textsc{$k$-Hypergraph Demand Matching} problem, there is a combinatorial $(k, 1)$-bicriteria approximation algorithm with respect to \eqref{eq:hdm-lp}. In particular, there is a combinatorial $(2, 1)$-bicriteria approximation algorithm for the \textsc{Demand Matching} problem with respect to \eqref{eq:lp}.
\end{theorem}

The algorithm in \Cref{thm:main} uses a simple greedy strategy. Informally, the algorithm scans the hyperedges in a nonincreasing order of \emph{density}, namely $w(e)/d(e)$ for a hyperedge $e \in \mathcal E$, and adds a hyperedge whenever every vertex it contains has load at most its original capacity. This greedy framework is combinatorial and arguably quite simple, and its running time is dominated by the time required to sort the hyperedges.

This greedy algorithm is reminiscent of the greedy algorithm of \citet{dantzig1957discrete} for exactly solving the \textsc{Fractional Knapsack} problem, which is defined as follows. The input of the \textsc{Fractional Knapsack} problem consists of a \emph{capacity} $B \in \QQ_{\geq 0}$ and $n$ elements with \emph{demands} $d : [n] \to \QQ_{> 0}$ and \emph{weights} $w : [n] \to \QQ_{\geq 0}$. The objective is to choose a fraction $x_i \in [0, 1]$ of each element so that $\sum_{i = 1}^n x_i d(i) \leq B$ and $\sum_{i = 1}^n x_i w(i)$ is maximized. The greedy algorithm of \citet{dantzig1957discrete} sorts all elements in a nonincreasing order of \emph{density}, and while the total demand of chosen fractional elements is less than $B$, takes the largest possible fraction of each element in this order. It is not hard to observe that this greedy algorithm takes all but at most one selected element \emph{fully}. Therefore, one can obtain a $(1, 1)$-bicriteria approximation for the \textsc{Knapsack} problem by rounding the single fractional element, if any, up to a full element in the solution obtained by the greedy algorithm of \citet{dantzig1957discrete}.

We remark that our greedy algorithm is a variant of the greedy algorithm of \citet*{chekuri2009unsplittable} which yields a $k$-approximation guarantee for the unit-weight case of the \textsc{$k$-Hypergraph Demand Matching} problem, where all hyperedges have weight $1$. More generally, their algorithm gives a $k$-approximation guarantee for the unit-weight case of the \textsc{$k$-Sparse Column-Restricted Packing Integer Programming} problem.

\subsection{Further Relevant Work}

We briefly survey results on approximation algorithms for combinatorial optimization problems relevant to the \textsc{Demand Matching} problem. This subsection is not directly related to the main results of this paper and may therefore be skipped.

We begin with the \textsc{$k$-Column-Sparse Packing Integer Program} (\textsc{$k$-CS-PIP}) problem, which generalizes the \textsc{$k$-Hypergraph Demand Matching} problem. In the \textsc{$k$-CS-PIP} problem, every hyperedge $e \in \mathcal E$ is allowed to have different demand values $d^v(e)$ for all $v \in e$, and the output $\mathcal M \subseteq \mathcal E$ is required to satisfy $d^v(\mathcal M \cap \delta(v)) := \sum_{e \in \mathcal M \cap \delta(v)} d^v(e) \leq b(v)$ for all $v \in V$. The iterative packing approach of \citet{parekh2011iterative}, more generally, yields a $2k$-approximation (resp., $3$-approximation) for the special case of \textsc{$k$-CS-PIP} (resp., $2$-CS-PIP) where the hyperedges can be ordered, namely, $e_1, \ldots, e_m$, so that the demands are nondecreasing at each vertex (i.e., for all $i, j \in [m]$ with $i \leq j$ and $v \in e_i \cap e_j$, we have $d^v(e_i) \leq d^v(e_j)$).

\citet{ahmadian2017further} considered a generalization of \textsc{$2$-CS-PIP}, where an additional matroid constraint is present (i.e., a matroid on the edge set is given and the output is required to be independent in the matroid). They gave the aforementioned $(1, 2)$-bicriteria approximation algorithm for this generalized problem, which is used to obtain a $25/3$-approximation algorithm and improved approximation ratios for several special cases.

The \textsc{Demand Matching} problem is a special case of the \textsc{Multicommodity Demand Flow on a Tree} problem (also known as \textsc{All-or-Nothing Flow on a Tree} or \textsc{Unsplittable Flow on a Tree}), which we define as follows. The input consists of a tree $T = (V, E)$, edge capacities $b : E \to \QQ_{\geq 0}$, and a finite set $D$ of requests with demands $d : D \to \QQ_{> 0}$ and weights $w : D \to \QQ_{\geq 0}$. Each request $i \in D$ wishes to route $d(i)$ units of flow from $s_i \in V$ to $t_i \in V$ along the unique $s_i$--$t_i$ path on $T$. The objective is to find a maximum weight subset of requests satisfying the edge capacities. The \textsc{Demand Matching} problem can be interpreted as the special case of this problem where $T$ is a star.

The \textsc{Multicommodity Demand Flow on a Tree} problem has been extensively studied. Under the assumption that the maximum demand is at most the minimum capacity, \citet*{chekuri2007multicommodity} proved that the integrality gap of the natural LP relaxation is at most $4$ if all demands are equal to one and at most $46.168$ if demands are arbitrary. (Without this assumption, the integrality gap is superconstant for arbitrary demands.) For unit demands, \citet*{konemann2014multicommodity} gave a $(1, 2)$-bicriteria approximation algorithm using the iterative relaxation framework, which implies a $\min\{ 3, 1 + O(1/b_{\min}) \}$-approximation algorithm for this problem and several variants, where $b_{\min}$ denotes the minimum capacity.

The special case of the \textsc{Multicommodity Demand Flow on a Tree} problem where $T$ is a path, as well as its generalization in which $T$ is replaced by a general graph, has also attracted much attention. For conciseness, we omit a discussion of these variants.

A further generalization of the \textsc{Multicommodity Demand Flow on a Tree} problem is called the \textsc{Column-Restricted Packing Integer Program} (\textsc{CR-PIP}) problem, introduced by \citet{kolliopoulos2001approximation}. A \emph{column-restricted packing integer program} is an integer program of the form
\begin{equation} \label{eq:cr-pip}
    \max\left\{ w^{\T} x : A \diag(d) x \leq b, x \in \{ 0, 1 \}^n \right\},
\end{equation}
where $A$ is an $m \times n$ $\{ 0, 1 \}$-valued matrix, $b \in \ZZ_+^m$, $d, w \in \ZZ_+^n$, and $\diag(d)$ denotes the $n \times n$ diagonal matrix whose diagonal entries are the components of $d$. It is interesting to ask how the integrality gap of the natural LP relaxation of \eqref{eq:cr-pip} is related to that of its unit-demand counterpart, i.e.,
\begin{equation} \label{eq:cr-pip-unit}
    \max\left\{ w^{\T} x : Ax \leq b, x \in \{ 0, 1 \}^n \right\}.
\end{equation}
Building upon the grouping and scaling technique of \citet{kolliopoulos2001approximation}, \citet*{chekuri2007multicommodity} proved that, if the maximum demand is at most the minimum capacity, then the integrality gap of the natural LP relaxation of \eqref{eq:cr-pip} is at most $11.542$ times that of \eqref{eq:cr-pip-unit}.

\subsection{Organization of This Paper}

This paper is organized as follows. In \cref{sec:cycle}, we prove properties of the \textsc{Demand Matching} problem on a cycle with the no-bottleneck assumption, which are used in subsequent sections. In \cref{sec:better2}, we present and analyze the iterative relaxation algorithm for \Cref{thm:better2}. In \cref{sec:parametric}, we generalize \Cref{thm:better2} to a family of parametric bicriteria approximation algorithms, proving \Cref{thm:parametric-nba}. In \cref{sec:greedy}, we present and analyze the greedy algorithm for \Cref{thm:main}. In \cref{sec:tight}, we prove the tightness result in \Cref{thm:tight-nba}.
\section{Demand Matching on a Cycle} \label{sec:cycle}

In this section, we study the special case of the \textsc{Demand Matching} problem where the underlying graph is a cycle and the no-bottleneck assumption is satisfied. First, we observe that this special case reduces to a \textsc{Maximum Weight $b$-Matching} problem on the cycle.

\begin{lemma} \label{lem:dm-cycle}
    Let $(C = (V, E), b, d, w)$ be an instance of the \textsc{Demand Matching} problem such that $C$ is a cycle and that the no-bottleneck assumption is satisfied. Let $\hat b : V \to \RR \cup \{ \infty \}$ be defined by
    \begin{equation} \label{eq:hat-b}
        \hat b(v) := \left\{
        \begin{array}{ll}
            1 & \text{if $d(\delta(v)) > b(v)$}, \\
            \infty & \text{otherwise}.
        \end{array}
    \right.
    \end{equation}
    Then a subset $M \subseteq E$ is a $(b, d)$-feasible demand matching if and only if $M$ is a $\hat b$-matching.
\end{lemma}

\begin{proof}
    Every vertex $v$ of a cycle has exactly two incident edges. By the no-bottleneck assumption, either incident edge is feasible at $v$, so the capacity constraint at $v$ is not satisfied exactly when both incident edges are contained and $d(\delta(v)) > b(v)$.
\end{proof}

\begin{lemma} \label{lem:dm-cycle-polytime}
    There is a polynomial-time algorithm for the special case of the \textsc{Demand Matching} problem in which the input graph is a cycle and the no-bottleneck assumption is satisfied.
\end{lemma}

\begin{proof}
    By \Cref{lem:dm-cycle}, the problem reduces to finding a maximum weight $\hat b$-matching, where $\hat b$ is defined in \eqref{eq:hat-b}; this can be done in polynomial time \cite{marsh1979matching,padberg1982odd,gabow1983efficient,anstee1987polynomial,gerards1995matching}.
\end{proof}

The following lemma relates three quantities defined on a cycle: the optimum of \eqref{eq:lp}, the total edge weight, and the maximum weight of a demand matching with \emph{relaxed} vertex capacities. This lemma is the crux for the weight guarantee of \Cref{alg:better2}. In particular, we relate the optimum of \eqref{eq:lp} to that of the fractional matching LP, and exploit the half-integrality of extreme points of the fractional matching polytope.

\begin{lemma} \label{lem:dm-cycle-lp}
    Let $C = (V, E)$ be a cycle. Let $b : V \to \QQ_{\geq 0}$. Let $d : E \to \QQ_{> 0}$ and $w : E \to \QQ_{\geq 0}$. Let $\Delta \geq d_{\max} := \max_{e \in E} d(e)$. Let $x^\star \in \RR^E$ be an extreme point optimal solution to \eqref{eq:lp} on the instance $(C, b, d, w)$. Let $M^\star \subseteq E$ be a maximum weight (with respect to $w$) $(b + \Delta, d)$-feasible demand matching in $C$. Then
    $$ \sum_{e \in E} x_e^\star w(e) \leq w(M^\star) + \frac{w(E)}{6}. $$
\end{lemma}

\begin{proof}
    Note that an instance with vertex capacities $b + \Delta$ and edge demands $d$ always satisfies the no-bottleneck assumption since $d(e) \leq d_{\max} \leq b(v) + d_{\max} \leq b(v) + \Delta$ for all $v \in V$ and all $e \in \delta(v)$. By \Cref{lem:dm-cycle}, a subset $M \subseteq E$ is a $(b + \Delta, d)$-feasible demand matching if and only if $M$ is a $\hat b$-matching, where $\hat b : V \to \QQ_{\geq 0} \cup \{ \infty \}$ is defined by
    $$ \hat b(v) := \left\{
        \begin{array}{ll}
            1 & \text{if $d(\delta(v)) > b(v) + \Delta$}, \\
            \infty & \text{otherwise}.
        \end{array}
    \right. $$
    
    Let $V_1 := \{ v \in V : \hat b(v) = 1 \}$. Let $v \in V_1$. Then $|\delta(v)| = 2$ and $d(\delta(v)) > b(v) + \Delta$. Let $e_1$ and $e_2$ be the two edges incident to $v$. Since $x^\star$ is feasible in \eqref{eq:lp} on the instance $(C, b, d, w)$,
    $$ d(\delta(v)) > b(v) + \Delta \geq \sum_{e \in \delta(v)} x_e^\star d(e) + \Delta. $$
    Rearranging this inequality yields that
    $$ \Delta \sum_{e \in \delta(v)} (1 - x_e^\star) \geq \sum_{e \in \delta(v)} (1 - x_e^\star) d(e) = d(\delta(v)) - \sum_{e \in \delta(v)} x_e^\star d(e) > \Delta. $$
    Hence,
    $$ 2 - x_{e_1}^\star - x_{e_2}^\star = (1 - x_{e_1}^\star) + (1 - x_{e_2}^\star) = \sum_{e \in \delta(v)} (1 - x_e^\star) > 1. $$
    It follows that for all $v \in V_1$,
    $$ x^\star(\delta(v)) = x_{e_1}^\star + x_{e_2}^\star < 1. $$
    
    Let $y^\star \in \RR^E$ be an extreme point optimal solution to the following LP:
    \begin{alignat}{4}
        \maximize \qquad && \sum_{e \in E} y_e w(e) \tag{FM-LP} \label{eq:fm-lp} \\
        \subjectto \qquad && \sum_{e \in \delta(v)} y_e & \leq 1 && \qquad \forall v \in V_1, \notag \\
        && y_e &\in [0, 1] && \qquad \forall e \in E. \notag
    \end{alignat}
    Since $x^\star$ is feasible in \eqref{eq:fm-lp}, we have
    $$ \sum_{e \in E} x_e^\star w(e) \leq \sum_{e \in E} y_e^\star w(e). $$
    
    We claim that
    \begin{equation} \label{eq:fm-ext-pt}
        \sum_{e \in E} y_e^\star w(e) \leq \max \left\{ w(M^\star), \frac{w(E)}{2} \right\}.
    \end{equation}
    We have two cases. First, suppose that $V_1 = V$. Then \eqref{eq:fm-lp} is the fractional matching LP on $C$. It is known that $y^\star$ is half-integral \cite{balinski1965integer}. Hence, either $y^\star$ is integral, or $y_e^\star = 1/2$ for all $e \in E$. If $y^\star$ is integral, then $\supp(y^\star)$ is a matching (and thus a $\hat b$-matching) in $C$. By \Cref{lem:dm-cycle}, $\supp(y^\star)$ is a $(b + \Delta, d)$-feasible demand matching in $C$. Since $M^\star$ is a maximum weight $(b + \Delta, d)$-feasible demand matching in $C$, we have $\sum_{e \in E} y_e^\star w(e) = w(\supp(y^\star)) \leq w(M^\star)$. If $y_e^\star = 1/2$ for all $e \in E$, we have $\sum_{e \in E} y_e^\star w(e) = w(E)/2$. Hence, we have \eqref{eq:fm-ext-pt}. Second, suppose that $V_1 \subsetneq V$. Then \eqref{eq:fm-lp} is the maximum weight fractional matching LP on a vertex-disjoint union of paths. It is known that $y^\star$ is integral \cite{balinski1965integer}. By the same argument as in the first case,
    $$ \sum_{e \in E} y_e^\star w(e) \leq w(M^\star). $$
    Combining the two cases proves \eqref{eq:fm-ext-pt}.
    
    It is not hard to see that $C$ can be partitioned into three matchings; indeed, alternating edges form two matchings, with the remaining edge forming a third when $C$ is odd. Hence, at least one of these three matchings has weight at least $w(E)/3$. Since $M^\star$ is a maximum weight $(b + \Delta, d)$-feasible demand matching on $C$, it is a maximum weight $\hat b$-matching on $C$ by \Cref{lem:dm-cycle}. Since $\hat b(v) \geq 1$ for all $v \in V$, $M^\star$ has weight at least that of any matching on $C$. It follows that
    $$ w(M^\star) \geq \frac{w(E)}{3}. $$
    We claim that $\sum_{e \in E} x_e^\star w(e) \leq w(M^\star) + w(E)/6$. If $\sum_{e \in E} x_e^\star w(e) \leq w(M^\star)$, then we are done. Otherwise,
    $$ \sum_{e \in E} x_e^\star w(e) \leq \frac{w(E)}{2} = \frac{w(E)}{3} + \frac{w(E)}{6} \leq w(M^\star) + \frac{w(E)}{6}. $$
    This completes the proof.
\end{proof}
\section{Iterative Relaxation Better-of-Two Algorithm} \label{sec:better2}

In this section, we present the algorithm for \Cref{thm:better2}. As discussed earlier, our algorithm for \Cref{thm:better2} uses iterative relaxation techniques. Informally, our algorithm repeatedly solves the natural LP relaxation of a \emph{residual} instance and obtains an extreme point optimal solution. In each iteration, one of the following cases must occur: the solution has a newly integral edge, which we freeze at its current value; there is a vertex of degree at most one in the residual instance, in which case we drop the corresponding LP constraint; or, as we show in \Cref{lem:extreme-point}, the support graph of the solution is a vertex-disjoint union of odd cycles. In the last case, we take the better of the following two rounding strategies:
\begin{itemize}
    \item[(S1)] Take a maximum weight demand matching on each odd cycle, with each vertex capacity being the residual capacity \emph{relaxed by $d_{\max}$}.
    \item[(S2)] Take all edges of the odd cycles, and discard all previously fixed edges.
\end{itemize}
As we argue in \Cref{lem:dm-cycle-polytime}, a maximum weight demand matching on a cycle can be computed in polynomial time.

Now, we define the LP used in the algorithm. For all $W \subseteq V$, $F \subseteq E$, and $\tilde b : V \to \QQ$ that is nonnegative on $W$, we use $\LP(W, F, \tilde b)$ to denote the following LP:
\begin{alignat*}{4}
    \maximize \qquad && \sum_{e \in F} x_e w(e) \\
    \subjectto \qquad && \sum_{e \in F \cap \delta(v)} x_e d(e) & \leq \tilde b(v) && \qquad \forall v \in W \\
    && x_e &\in [0, 1] && \qquad \forall e \in F.
\end{alignat*}
We show, by a simple counting argument, the following structural lemma for a strictly fractional extreme point of $\LP(W, F, \tilde b)$. We note that this structure is also stated in \cite[Lemma 3.1]{shepherd2007demand} with an equivalent yet slightly different formulation; we state and prove it here for the sake of completeness.

\begin{lemma}[\citet{shepherd2007demand}] \label{lem:extreme-point}
    Let $W \subseteq V$ and $F \subseteq E$. Let $\tilde b : V \to \QQ$ be nonnegative on $W$. Let $x^\star \in \RR^F$ be an extreme point of $\LP(W, F, \tilde b)$ such that $0 < x_e^\star < 1$ for all $e \in F$. Suppose that $|F \cap \delta(v)| \geq 2$ for all $v \in W$. Then $(W, F)$ is a vertex-disjoint union of odd cycles, and $V \setminus W$ is a set of isolated vertices.
\end{lemma}

\begin{proof}
    Since $0 < x_e^\star < 1$ for all $e \in F$, a rank argument (see, e.g., Theorem 5.7 of \cite{schrijver2003combinatorial}) implies that $|W| \geq |F|$. Then
    $$ 2|W| \leq \sum_{v \in W} |F \cap \delta(v)| \leq \sum_{v \in V} |F \cap \delta(v)| = 2|F| \leq 2|W|. $$
    Hence, we have $|F \cap \delta(v)| = 2$ for all $v \in W$, $|F \cap \delta(v)| = 0$ for all $v \in V \setminus W$, and $|F| = |W|$. It follows that $(W, F)$ is a vertex-disjoint union of cycles, and $V \setminus W$ is a set of isolated vertices.

    If we consider the full-rank constraint matrix of $\LP(W, F, \tilde b)$ whose columns are indexed by $F$ and rows indexed by $W$, and scale each column $e \in F$ by $1/d(e)$, we see that any even cycle $C = (V_C, E_C)$ would make this matrix singular since the rows in $V_C$ are not linearly independent. This proves that $(W, F)$ is a vertex-disjoint union of odd cycles.
\end{proof}

Now, we formally state the algorithm in \Cref{alg:better2}.

\begin{algorithm}[ht]
    \caption{A better-of-two algorithm for the \textsc{Demand Matching} problem.}
    \label{alg:better2}
    \KwIn{a graph $G = (V, E)$, $b : V \to \QQ_{\geq 0}$, $d : E \to \QQ_{> 0}$, and $w : E \to \QQ_{\geq 0}$, such that the no-bottleneck assumption is satisfied.}
    \KwOut{$M \subseteq E$.}
    $M_1 \leftarrow \emptyset$, $M_2 \leftarrow \emptyset$, $W \leftarrow V$, $F \leftarrow E$, $\tilde b \leftarrow b$, $d_{\max} \leftarrow \max_{e \in E} d(e)$. \\
    \While{$F \neq \emptyset$}{
        Let $x^\star \in \RR^F$ be an extreme point optimal solution to $\LP(W, F, \tilde b)$. \\
        \uIf{$x_e^\star = 0$ for some $e \in F$}{
            $F \leftarrow F \setminus \{ e \}$. \CommentSty{\textcolor{blue}{// Case 1}}
        }
        \uElseIf{$x_e^\star = 1$ for some $e = uv \in F$}{
            $M_1 \leftarrow M_1 \cup \{ e \}$, $F \leftarrow F \setminus \{ e \}$, $\tilde b(u) \leftarrow \tilde b(u) - d(e)$, $\tilde b(v) \leftarrow \tilde b(v) - d(e)$. \CommentSty{\textcolor{blue}{// Case 2}}
        }
        \uElseIf{$|F \cap \delta(v)| \leq 1$ for some $v \in W$}{
            $W \leftarrow W \setminus \{ v \}$. \CommentSty{\textcolor{blue}{// Case 3}}
        }
        \Else{
            \CommentSty{\textcolor{blue}{// Case 4: $0 < x_e^\star < 1$ for all $e \in F$ and $|F \cap \delta(v)| \geq 2$ for all $v \in W$}} \\
            \Cref{lem:extreme-point} implies that $(W, F)$ is a vertex-disjoint union $\mathcal C$ of odd cycles. \\
            \ForEach{$C \in \mathcal C$}{
                Find a maximum weight $(\tilde b + d_{\max}, d)$-feasible demand matching $M_1^C$ in $C$. \\
                $M_1 \leftarrow M_1 \cup M_1^C$.
            }
            $M_2 \leftarrow \bigcup_{(V_C, E_C) \in \mathcal C} E_C$. \\
            $F \leftarrow \emptyset$.
        }
    }
    \Return{$\argmax_{M \in \{ M_1, M_2 \}} w(M)$}
\end{algorithm}

We analyze \Cref{alg:better2}. Consider the execution of \Cref{alg:better2} on input $(G = (V, E), b, d, w)$. Without loss of generality, we assume that the execution has at least one iteration. Let $T \in \NN$ be the number of iterations. For $i \in [T]$, we use $M_1^i$, $W_i$, $F_i$, $\tilde b_i$, and $x^{(i)}$ to denote the values of $M_1$, $W$, $F$, $\tilde b$, and $x^\star$, respectively, at the \emph{end} of the $i^\text{th}$ iteration. Moreover, we use $M_1^0$, $W_0$, $F_0$, and $\tilde b_0$ to denote the corresponding values before the first iteration. For $i \in \{ 0, \ldots, T \}$, we use $\LPOPT_i$ to denote the optimum of $\LP(W_i, F_i, \tilde b_i)$. Note that $x^{(i)} \in \RR^{F_{i - 1}}$ is an optimal extreme point solution to $\LP(W_{i - 1}, F_{i - 1}, \tilde b_{i - 1})$ for $i \in [T]$ by these definitions.

\begin{lemma} \label{lem:better2-polytime}
    \Cref{alg:better2} terminates in polynomial time.
\end{lemma}

\begin{proof}
    In each iteration, \Cref{alg:better2} either removes at least one edge from $F \subseteq E$, or removes a vertex from $W \subseteq V$. Hence, \Cref{alg:better2} terminates in at most $|V| + |E|$ iterations.

    Consider an odd cycle $C = (V_C, E_C)$ in case 4. Since $V_C \subseteq W$, we have that $\tilde b$ is nonnegative on $V_C$. Hence, the \emph{relaxed} instance restricted to $C$ with vertex capacities $\tilde b|_{V_C} + d_{\max}$ and edge demands $d|_{E_C}$ satisfies the no-bottleneck assumption. By \Cref{lem:dm-cycle-polytime}, finding a maximum weight $(\tilde b + d_{\max}, d)$-feasible demand matching in $C$ can be done in polynomial time. Therefore, each iteration runs in polynomial time. This completes the proof.
\end{proof}

\begin{lemma} \label{lem:better2-capacity}
    Let $M \subseteq E$ be the output of \Cref{alg:better2}. Then $d(M \cap \delta(v)) \leq b(v) + d_{\max}$ for all $v \in V$.
\end{lemma}

\begin{proof}
    Let $v \in V$. Let $M_1, M_2 \subseteq E$ be the results of the two rounding strategies of \Cref{alg:better2}, respectively. First, we consider $M_2$; suppose without loss of generality that $M_2 \neq \emptyset$. If $v$ does not belong to any cycle in the vertex-disjoint union $\mathcal C$ guaranteed by \Cref{lem:extreme-point} in case 4, then $M_2 \cap \delta(v) = \emptyset$. Otherwise, let $e, f \in E$ be the two edges incident to $v$ in some odd cycle $C \in \mathcal C$, and the no-bottleneck assumption implies that $\max\{ d(e), d(f) \} \leq b(v)$; hence,
    $$ d(M_2 \cap \delta(v)) = d(e) + d(f) \leq b(v) + d_{\max}. $$
    
    Now, we consider $M_1$. First, suppose that $v$ is removed from $W$ in iteration $j$ of \Cref{alg:better2}. It is not hard to show by induction that for all $i \in \{ 0, \ldots, j - 1\}$,
    \begin{equation} \label{eq:induction}
        d(M_1^i \cap \delta(v)) + \tilde b_i(v) = b(v).
    \end{equation}
    Let $x \in \RR^{F_{j - 1}}$ be a feasible solution to $\LP(W_{j - 1}, F_{j - 1}, \tilde b_{j - 1})$. Then $0 \leq x_e \leq 1$, and
    $$ \sum_{e \in F_{j - 1} \cap \delta(v)} x_e d(e) \leq \tilde b_{j - 1}(v) = b(v) - d(M_1^{j - 1} \cap \delta(v)). $$
    Since $v$ is removed from $W$ in the $j^\text{th}$ iteration, we have $|F_{j - 1} \cap \delta(v)| \leq 1$. Since $M_1 \subseteq M_1^{j - 1} \cup F_{j - 1}$,
    \begin{multline*}
        d(M_1 \cap \delta(v)) \leq d(M_1^{j - 1} \cap \delta(v)) + d(F_{j - 1} \cap \delta(v)) \leq b(v) - \sum_{e \in F_{j - 1} \cap \delta(v)} x_e d(e) + d(F_{j - 1} \cap \delta(v)) \\
        = b(v) + \sum_{e \in F_{j - 1} \cap \delta(v)} \left(1 - x_e\right) d(e) \leq b(v) + |F_{j - 1} \cap \delta(v)| \cdot 1 \cdot d_{\max} \leq b(v) + d_{\max}.
    \end{multline*}

    Suppose that $v$ is never removed from $W$ in \Cref{alg:better2} and that, if case 4 occurs, $v$ does not belong to any cycle in the vertex-disjoint union $\mathcal C$ guaranteed by \Cref{lem:extreme-point} in case 4. Then \eqref{eq:induction} holds for all $i \in \{ 0, \ldots, T \}$. Hence,
    $$ d(M_1 \cap \delta(v)) = d(M_1^T \cap \delta(v)) = b(v) - \tilde b_T(v) \leq b(v). $$
    It remains to consider the case where $v$ belongs to an odd cycle $C$ in the vertex-disjoint union $\mathcal C$ guaranteed by \Cref{lem:extreme-point} in case 4 in iteration $T$. Then \eqref{eq:induction} holds for all $i \in \{ 0, \ldots, T - 1 \}$. Since we take a maximum weight $(\tilde b_T + d_{\max}, d)$-feasible demand matching $M_1^C$ in $C$, we have $d(M_1^C \cap \delta(v)) \leq \tilde b_T(v) + d_{\max}$. Therefore,
    $$ d(M_1 \cap \delta(v)) = d(M_1^{T - 1} \cap \delta(v)) + d(M_1^C \cap \delta(v)) \leq b(v) - \tilde b_{T - 1}(v) + \tilde b_T(v) + d_{\max} = b(v) + d_{\max}, $$
    where the last equality uses $\tilde b_{T - 1}(v) = \tilde b_T(v)$. Since $M \in \{ M_1, M_2 \}$, this completes the proof.
\end{proof}

\begin{lemma} \label{lem:better2-weight}
    Let $M \subseteq E$ be the output of \Cref{alg:better2}. Then $w(M)$ is at least $6/7$ times the optimum of \eqref{eq:lp}. If $G$ is bipartite, then $w(M)$ is at least the optimum of \eqref{eq:lp}.
\end{lemma}

\begin{proof}
    It is not hard to see that for every iteration indexed by $i \in [T]$ in which one of cases 1 to 3 occurs,
    $$ w(M_1^{i - 1}) + \LPOPT_{i - 1} \leq w(M_1^i) + \LPOPT_i. $$
    Indeed,
    \begin{itemize}[itemsep=0pt]
        \item deleting an edge $e$ with $x_e = 0$ preserves the current LP solution and optimum;
        \item fixing an edge $e$ with $x_e = 1$ decreases the residual LP objective by $w(e)$, while increasing $w(M_1)$ by $w(e)$;
        \item dropping a constraint enlarges the feasible region and therefore cannot decrease the residual LP optimum.
    \end{itemize}
    
    If case 4 never occurs, then $w(M)$ is at least the optimum of \eqref{eq:lp} by noting that $M_1^0 = \emptyset$, that $w(M_1^T) = w(M)$, that $\LPOPT_T = 0$, and that $\LPOPT_0$ is equal to the optimum of \eqref{eq:lp}. If $G$ is bipartite, then there is no odd cycle in $G$, so case 4 never occurs.

    Suppose that case 4 occurs in iteration $T$. Let $M_1, M_2 \subseteq E$ be the subsets produced by the two rounding strategies of \Cref{alg:better2}, respectively. Let $\mathcal C$ be the vertex-disjoint union of odd cycles guaranteed by \Cref{lem:extreme-point}. Let $\LPOPT$ be the optimum of \eqref{eq:lp}. Then
    $$ \LPOPT \leq w(M_1^{T - 1}) + \sum_{e \in F_{T - 1}} x_e^{(T)} w(e). $$
    For $C \in \mathcal C$, let $M_1^C$ be a maximum weight $(\tilde b_{T - 1} + d_{\max}, d)$-feasible demand matching in $C$. Then
    $$ M_1 = M_1^{T - 1} \cup \bigcup_{C \in \mathcal C} M_1^C. $$
    For each $C = (V_C, E_C) \in \mathcal C$, since $V_C \subseteq W_{T - 1}$, we have that $\tilde b_{T - 1}$ is nonnegative on $V_C$. Since the odd cycles in $\mathcal C$ are vertex-disjoint and since the residual LP decomposes over the odd cycles in $\mathcal C$, the restriction of $x^{(T)}$ to each odd cycle $C = (V_C, E_C) \in \mathcal C$ is an extreme point optimal solution to \eqref{eq:lp} on the instance $(C, \tilde b_{T - 1}|_{V_C}, d|_{E_C}, w|_{E_C})$. Applying \Cref{lem:dm-cycle-lp} to each odd cycle in $\mathcal C$ yields that
    \begin{align*}
        \LPOPT &\leq w(M_1^{T - 1}) + \sum_{e \in F_{T - 1}} x_e^{(T)} w(e) = w(M_1^{T - 1}) + \sum_{(V_C, E_C) \in \mathcal C} \sum_{e \in E_C} x_e^{(T)} w(e) \\
        &\leq w(M_1^{T - 1}) + \sum_{C = (V_C, E_C) \in \mathcal C} \left(w(M_1^C) + \frac{w(E_C)}{6}\right) \\
        &= w(M_1^{T - 1}) + \sum_{C \in \mathcal C} w(M_1^C) + \frac{1}{6} \sum_{(V_C, E_C) \in \mathcal C} w(E_C) \\
        &= w(M_1) + \frac{w(F_{T - 1})}{6}.
    \end{align*}
    On the other hand,
    $$ w(M_2) = \sum_{(V_C, E_C) \in \mathcal C} w(E_C) = w(F_{T - 1}). $$
    Since $w(M) = \max \{ w(M_1), w(M_2) \}$,
    $$ \LPOPT \leq w(M_1) + \frac{w(M_2)}{6} \leq w(M) + \frac{w(M)}{6} = \frac{7}{6} \cdot w(M). $$
    This completes the proof.
\end{proof}

\Cref{lem:better2-polytime,lem:better2-capacity,lem:better2-weight} together prove \Cref{thm:better2}.
\section{Parametric Bicriteria Guarantees} \label{sec:parametric}

In this section, we generalize the $(7/6, 1)$-bicriteria approximation algorithm from \cref{sec:better2} to provide parametric bicriteria guarantees $(\alpha^\mathrm{UB}(\beta), \beta)$ for $\beta \geq 1$, where $\alpha^\mathrm{UB}$ is defined in \eqref{eq:alpha-ub}. In particular, this gives a $(1, 4/3)$-bicriteria approximation algorithm.

We state the algorithm in \Cref{alg:parametric-nba}. This algorithm is identical to \Cref{alg:better2} except for the handling of odd cycles, where \Cref{alg:parametric-nba} takes a maximum weight $(\tilde b + \beta d_{\max}, d)$-feasible demand matching in each odd cycle for the first strategy and \Cref{alg:better2} is the special case where $\beta = 1$.

\begin{algorithm}
    \caption{A parametric bicriteria approximation algorithm for \textsc{Demand Matching}.}
    \label{alg:parametric-nba}
    \KwIn{a graph $G = (V, E)$, $b : V \to \QQ_{\geq 0}$, $d : E \to \QQ_{> 0}$, $w : E \to \QQ_{\geq 0}$, and $\beta \geq 1$, such that the no-bottleneck assumption is satisfied.}
    \KwOut{$M \subseteq E$.}
    $M_1 \leftarrow \emptyset$, $M_2 \leftarrow \emptyset$, $W \leftarrow V$, $F \leftarrow E$, $\tilde b \leftarrow b$, $d_{\max} \leftarrow \max_{e \in E} d(e)$. \\
    \While{$F \neq \emptyset$}{
        Let $x^\star \in \RR^F$ be an extreme point optimal solution to $\LP(W, F, \tilde b)$. \\
        \uIf{$x_e^\star = 0$ for some $e \in F$}{
            $F \leftarrow F \setminus \{ e \}$. \CommentSty{\textcolor{blue}{// Case 1}}
        }
        \uElseIf{$x_e^\star = 1$ for some $e = uv \in F$}{
            $M_1 \leftarrow M_1 \cup \{ e \}$, $F \leftarrow F \setminus \{ e \}$, $\tilde b(u) \leftarrow \tilde b(u) - d(e)$, $\tilde b(v) \leftarrow \tilde b(v) - d(e)$. \CommentSty{\textcolor{blue}{// Case 2}}
        }
        \uElseIf{$|F \cap \delta(v)| \leq 1$ for some $v \in W$}{
            $W \leftarrow W \setminus \{ v \}$. \CommentSty{\textcolor{blue}{// Case 3}}
        }
        \Else{
            \CommentSty{\textcolor{blue}{// Case 4: $0 < x_e^\star < 1$ for all $e \in F$ and $|F \cap \delta(v)| \geq 2$ for all $v \in W$}} \\
            \Cref{lem:extreme-point} implies that $(W, F)$ is a vertex-disjoint union $\mathcal C$ of odd cycles. \\
            \ForEach{$C \in \mathcal C$}{
                Find a maximum weight $(\tilde b + \beta d_{\max}, d)$-feasible demand matching $M_1^C$ in $C$. \\
                $M_1 \leftarrow M_1 \cup M_1^C$.
            }
            $M_2 \leftarrow \bigcup_{(V_C, E_C) \in \mathcal C} E_C$. \\
            $F \leftarrow \emptyset$.
        }
    }
    \Return{$\argmax_{M \in \{ M_1, M_2 \}} w(M)$}
\end{algorithm}

The analysis of the running time and the capacity guarantee of \Cref{alg:parametric-nba} follows from the same reasoning as in \Cref{lem:better2-polytime,lem:better2-capacity} after replacing $d_{\max}$ by $\beta d_{\max}$.

\begin{lemma} \label{lem:parametric-nba-polytime}
    \Cref{alg:parametric-nba} terminates in polynomial time.
\end{lemma}

\begin{lemma} \label{lem:parametric-nba-capacity}
    Let $M \subseteq E$ be the output of \Cref{alg:parametric-nba}. Then $d(M \cap \delta(v)) \leq b(v) + \beta d_{\max}$ for all $v \in V$.
\end{lemma}

For the weight guarantee, we present a parametric generalization of \Cref{lem:dm-cycle-lp}. The proof follows the same overall argument as that of \Cref{lem:dm-cycle-lp}; we provide a proof sketch and omit repetitive details.

\begin{lemma} \label{lem:parametric-nba-weight}
    Let $C = (V, E)$ be a cycle. Let $b : V \to \QQ_{\geq 0}$. Let $d : E \to \QQ_{> 0}$ and $w : E \to \QQ_{\geq 0}$. Let $\Delta \geq d_{\max} := \max_{e \in E} d(e)$. Let $\beta \geq 1$. Let $x^\star \in \RR^E$ be an extreme point optimal solution to \eqref{eq:lp} on the instance $(C, b, d, w)$. Let $M^\star \subseteq E$ be a maximum weight (with respect to $w$) $(b + \beta \Delta, d)$-feasible demand matching in $C$. Let $\mu : [1, \infty) \to \RR$ be defined by
    \begin{equation} \label{eq:mu}
        \mu(\beta) := \left\{
            \begin{array}{ll}
                \frac{4 - 3\beta}{6} & \text{if $1 \leq \beta \leq 4/3$}, \\
                0 & \text{otherwise}.
            \end{array}
        \right.
    \end{equation}
    Then $\sum_{e \in E} x_e^\star w(e) \leq w(M^\star) + \mu(\beta) \cdot w(E)$.
\end{lemma}

\begin{proof}[Proof sketch]
    Note that an instance with vertex capacities $b + \beta \Delta$ and edge demands $d$ always satisfies the no-bottleneck assumption. By \Cref{lem:dm-cycle}, a subset $M \subseteq E$ is a $(b + \beta \Delta, d)$-feasible demand matching if and only if $M$ is a $\hat b$-matching, where $\hat b : V \to \QQ_{\geq 0} \cup \{ \infty \}$ is defined by
    $$ \hat b(v) := \left\{
        \begin{array}{ll}
            1 & \text{if $d(\delta(v)) > b(v) + \beta \Delta$}, \\
            \infty & \text{otherwise}.
        \end{array}
    \right. $$
    
    Let $V_1 := \{ v \in V : \hat b(v) = 1 \}$. Let $v \in V_1$. Then $|\delta(v)| = 2$ and $d(\delta(v)) > b(v) + \beta \Delta$. By the same argument as in \Cref{lem:dm-cycle-lp}, the feasibility of $x^\star$ in \eqref{eq:lp} on the instance $(C, b, d, w)$ implies that for all $v \in V_1$,
    $$ x^\star(\delta(v)) = x_{e_1}^\star + x_{e_2}^\star < 2 - \beta. $$

    We have two cases. First, suppose that $V_1 = V$. Then $\beta < 2$ and $x^\star / (2 - \beta)$ is feasible in \eqref{eq:fm-lp}. Let $\LPOPT := \sum_{e \in E} x_e^\star w(e)$. By the same argument as in \Cref{lem:dm-cycle-lp}, the half-integrality of extreme points of \eqref{eq:fm-lp} implies that
    $$ \frac{\LPOPT}{2 - \beta} \leq \max\left\{ w(M^\star), \frac{w(E)}{2} \right\}. $$
    If $w(M^\star) \geq w(E)/2$, then $\LPOPT \leq (2 - \beta) \cdot w(M^\star) \leq w(M^\star) \leq w(M^\star) + \mu(\beta) \cdot w(E)$, where the second inequality follows from the assumption that $\beta \geq 1$, and we are done. Otherwise,
    $$ \LPOPT \leq (2 - \beta) \cdot \frac{w(E)}{2} = \frac{w(E)}{3} + \frac{4 - 3\beta}{6} \cdot w(E) \leq w(M^\star) + \mu(\beta) \cdot w(E), $$
    where the second inequality follows from the facts that $w(M^\star) \geq w(E)/3$ and that $\mu(\beta) \geq (4-3\beta)/6$ for all $\beta \geq 1$. Second, suppose that $V_1 \subsetneq V$. Then \eqref{eq:fm-lp} is the maximum weight fractional matching LP on a vertex-disjoint union of paths. Since $\beta \geq 1$, we have $x^\star(\delta(v)) < 2 - \beta \leq 1$ for all $v \in V_1$, so $x^\star$ is feasible in \eqref{eq:fm-lp}. By the same argument as in \Cref{lem:dm-cycle-lp}, the integrality of extreme points of \eqref{eq:fm-lp} in this case implies that
    $$ \LPOPT \leq w(M^\star) \leq w(M^\star) + \mu(\beta) \cdot w(E), $$
    where the second inequality uses the fact that $\mu$ is nonnegative. This completes the proof.
\end{proof}

The weight guarantee of \Cref{alg:parametric-nba} follows from the same reasoning as in \Cref{lem:better2-weight}.

\begin{corollary} \label{cor:parametric-nba-weight}
    Let $M \subseteq E$ be the output of \Cref{alg:parametric-nba}. Then $w(M)$ is at least $1/(1 + \mu(\beta))$ times the optimum of \eqref{eq:lp}, where $\mu : [1, \infty) \to \RR$ is defined in \eqref{eq:mu}.
\end{corollary}

\begin{proof}[Proof sketch]
    Without loss of generality, we assume that case 4 occurs. Let $M_1, M_2 \subseteq E$ be the subsets produced by the two rounding strategies of \Cref{alg:parametric-nba}, respectively. Let $F', M_1' \subseteq E$ be the values of $F$ and $M_1$, respectively, at the beginning of the iteration handling case 4. Let $\mathcal C$ be the vertex-disjoint union of odd cycles guaranteed by \Cref{lem:extreme-point}. Following the argument in \Cref{lem:better2-weight}, \Cref{lem:parametric-nba-weight} applied to each odd cycle in $\mathcal C$ yields that
    \begin{multline*}
        \LPOPT \leq w(M_1') + \sum_{C = (V_C, E_C) \in \mathcal C} \left(w(M_1^C) + \mu(\beta) \cdot w(E_C)\right) = w(M_1) + \mu(\beta) \cdot w(F') \\
        = w(M_1) + \mu(\beta) \cdot w(M_2) \leq w(M) + \mu(\beta) \cdot w(M) = (1 + \mu(\beta)) \cdot w(M).
    \end{multline*}
    This completes the proof.
\end{proof}

\Cref{lem:parametric-nba-polytime,lem:parametric-nba-capacity} and \Cref{cor:parametric-nba-weight} together prove \Cref{thm:parametric-nba}.
\section{Greedy Algorithm} \label{sec:greedy}

In this section, we present a simple greedy algorithm for the \textsc{$k$-Hypergraph Demand Matching} problem. We state the algorithm in \Cref{alg:greedy-demand-matching}.

\begin{algorithm}
    \caption{A greedy algorithm for the \textsc{$k$-Hypergraph Demand Matching} problem.}
    \label{alg:greedy-demand-matching}
    \KwIn{a hypergraph $H = (V, \mathcal E)$, $b : V \to \QQ_{\geq 0}$, $d : \mathcal E \to \QQ_{> 0}$, and $w : \mathcal E \to \QQ_{\geq 0}$.}
    \KwOut{$\mathcal M \subseteq \mathcal E$.}
    $\mathcal M \leftarrow \emptyset$. \\
    Sort hyperedges $e \in \mathcal E$ in a nonincreasing order of $w(e) / d(e)$, breaking ties arbitrarily. \\
    \ForEach{$e \in \mathcal E$ in this order}{
        \If{$d(\mathcal M \cap \delta(v)) \leq b(v)$ for all $v \in e$}{
            $\mathcal M \leftarrow \mathcal M \cup \{ e \}$.
        }
    }
    \Return{$\mathcal M$}
\end{algorithm}

First, it is not hard to see that the subset $\mathcal M \subseteq \mathcal E$ output by \Cref{alg:greedy-demand-matching} satisfies the capacity guarantee.

\begin{lemma} \label{lem:capacity}
    The subset $\mathcal M \subseteq \mathcal E$ output by \Cref{alg:greedy-demand-matching} satisfies $d(\mathcal M \cap \delta(v)) \leq b(v) + d_{\max}$ for all $v \in V$.
\end{lemma}

\begin{proof}
    When a hyperedge $e \in \mathcal E$ is accepted, each $v \in e$ currently has load at most $b(v)$, so after adding $e$, its load is at most $b(v) + d(e) \leq b(v) + d_{\max}$; vertices not in $e$ are unchanged.
\end{proof}

We prove the weight guarantee for \Cref{alg:greedy-demand-matching}.

\begin{lemma} \label{lem:weight}
    Let $\mathcal M \subseteq \mathcal E$ be the output of \Cref{alg:greedy-demand-matching}. Let $x^\star \in \RR^{\mathcal E}$ be a feasible solution to \eqref{eq:hdm-lp}. Then
    $$ k \cdot w(\mathcal M) \geq \sum_{e \in \mathcal E} x_e^\star w(e). $$
\end{lemma}

\begin{proof}
    We say that the nonincreasing order of the hyperedges used in \Cref{alg:greedy-demand-matching} is the \emph{greedy order}. For each $e \in \mathcal E$, let $\rho(e) := w(e)/d(e)$, and let $\prev(e)$ be the set of hyperedges preceding $e$ in the greedy order. For each $e \in \mathcal R := \mathcal E \setminus \mathcal M$, fix $a(e) \in e$ so that, at the iteration where $e$ is considered,
    $$ d(\mathcal M \cap \delta(a(e)) \cap \prev(e)) > b(a(e)). $$
    For each $v \in V$, let
    $$ \mathcal R(v) := \{ e \in \delta(v) : a(e) = v \}. $$
    Let $V' := \{ v \in V : \mathcal R(v) \neq \emptyset \}$. For each $v \in V'$, let $e_v$ be the first hyperedge in $\mathcal R(v)$ in the greedy order, and let
    $$ \mathcal P(v) :=  \mathcal M \cap \delta(v) \cap \prev(e_v). $$
    Fix $v \in V'$. Since $x^\star$ is feasible in \eqref{eq:hdm-lp},
    $$ d(\mathcal P(v)) > b(v) \geq \sum_{e \in \delta(v)} x_e^\star d(e) \geq \sum_{e \in \mathcal R(v)} x_e^\star d(e) + \sum_{e \in \mathcal P(v)} x_e^\star d(e). $$
    Hence,
    $$ \sum_{e \in \mathcal R(v)} x_e^\star d(e) < \sum_{e \in \mathcal P(v)} (1 - x_e^\star) d(e). $$
    Since $\mathcal P(v) \subseteq \prev(e_v)$, we have $\rho(e) \geq \rho(e_v)$ for all $e \in \mathcal P(v)$. By the definition of $e_v$, we have $\rho(e) \leq \rho(e_v)$ for all $e \in \mathcal R(v)$. Since $\rho(e_v) = w(e_v)/d(e_v) \geq 0$,
    \begin{multline*}
        \sum_{e \in \mathcal R(v)} x_e^\star w(e) = \sum_{e \in \mathcal R(v)} x_e^\star \rho(e) d(e) \leq \rho(e_v) \sum_{e \in \mathcal R(v)} x_e^\star d(e) \\
        \leq \rho(e_v) \sum_{e \in \mathcal P(v)} (1 - x_e^\star) d(e) \leq \sum_{e \in \mathcal P(v)} (1 - x_e^\star) \rho(e) d(e) = \sum_{e \in \mathcal P(v)} (1 - x_e^\star) w(e).
    \end{multline*}
    Therefore,
    $$ \sum_{e \in \mathcal R} x_e^\star w(e) = \sum_{v \in V'} \sum_{e \in \mathcal R(v)} x_e^\star w(e) \leq \sum_{v \in V'} \sum_{e \in \mathcal P(v)} (1 - x_e^\star) w(e) \leq k\sum_{e \in \mathcal M} (1 - x_e^\star) w(e), $$
    where the last inequality follows from the fact that each hyperedge $e \in \mathcal M$ appears in $\mathcal P(v)$ for at most $k$ vertices $v \in V'$, namely the vertices contained in $e$. Hence,
    \begin{multline*}
        \sum_{e \in \mathcal E} x_e^\star w(e) = \sum_{e \in \mathcal R} x_e^\star w(e) + \sum_{e \in \mathcal M} x_e^\star w(e) \\
        \leq k\sum_{e \in \mathcal M} (1 - x_e^\star) w(e) + \sum_{e \in \mathcal M} x_e^\star w(e) = \sum_{e \in \mathcal M} (k - (k - 1) x_e^\star) w(e) \leq k \cdot w(\mathcal M).
    \end{multline*}
    This completes the proof.
\end{proof}

\Cref{lem:capacity,lem:weight} together prove \Cref{thm:main}.
\section{Tightness Result} \label{sec:tight}

In this section, we prove \Cref{thm:tight-nba}. We consider $\beta \in [0, 1)$ in \Cref{lem:tight-general-2} and $\beta \in [1, 4/3]$ in \Cref{lem:tight-nba}.

\begin{lemma} \label{lem:tight-general-2}
    Let $\beta \in [0, 1)$. Let $\varepsilon > 0$. Then there is an instance $(G, b, d, w)$ of the \textsc{Demand Matching} problem such that $w(M^\star)$ is at most $(2/(3(2 - \beta)) + \varepsilon)$ times the optimum of \eqref{eq:lp}, where $M^\star$ is an optimal $(b + \beta d_{\max}, d)$-feasible demand matching. Moreover, this instance satisfies the no-bottleneck assumption.
\end{lemma}

\begin{proof}
    Let $G = (V, E)$ be the triangle, i.e., $V := \{ v_1, v_2, v_3 \}$ and $E := \{ v_1 v_2, v_2 v_3, v_3 v_1 \} $. Let $D \in \NN$ be such that $(2 - \beta) D \geq 4$, that $8/(3D(2 - \beta)^2) \leq \varepsilon$, and that $D \leq (2 - \beta) D - 2$. Let $B := \lfloor (2 - \beta) D \rfloor - 1$. Let $b(v) := B$ for all $v \in V$. Let $d(e) := D$ and $w(e) := 1$ for all $e \in E$. Then $d_{\max} = D$, and $(G, b, d, w)$ satisfies the no-bottleneck assumption. We illustrate this instance in \Cref{fig:tight-general-2}.

    \begin{figure}[ht]
        \centering

        \begin{tikzpicture}[
            scale=1,
            vertex/.style={
                circle,
                draw,
                fill=white,
                inner sep=1.5pt,
                minimum size=22pt,
                font=\small
            },
            heavy/.style={thick},
            light/.style={thick, dashed}
        ]
        
        \node[vertex] (v1) at (0,0) {$B$};
        \node[vertex] (v2) at (4,0) {$B$};
        \node[vertex] (v3) at (2,3) {$B$};
        
        \draw[heavy] (v1) -- (v2);
        \draw[heavy] (v2) -- (v3);
        \draw[heavy] (v3) -- (v1);
        
        \begin{scope}[shift={(5,2.65)}]
            \node[inner sep=6pt, align=left, font=\small] at (0,0) {
                \begin{tikzpicture}[baseline=-0.5ex]
                    \draw[heavy] (0,0) -- (0.75,0);
                    \node[right] at (0.85,0) {$(d,w)=(D,1)$};
                    \node at (1.55, -.45) {$B := \lfloor (2 - \beta) D \rfloor - 1$};
                \end{tikzpicture}
            };
        \end{scope}
        
        \end{tikzpicture}
        
        \caption{The instance in the proof of \Cref{lem:tight-general-2}.}
        \label{fig:tight-general-2}
    \end{figure}
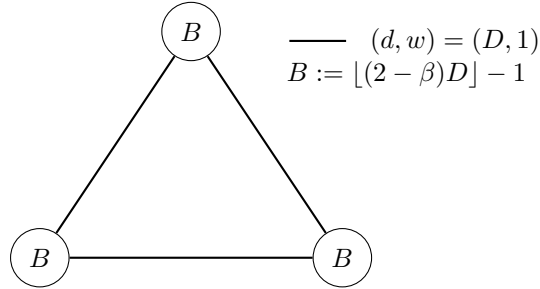

    Let $x^\star \in \RR^E$ be defined by $x_e^\star := B/(2D)$ for all $e \in E$. Since $\sum_{e \in \delta(v)} x_e^\star d(e) = 2 \cdot B/(2D) \cdot D = B = b(v)$ for all $v \in V$, we have that $x^\star$ is feasible in \eqref{eq:lp}.

    In any $(b + \beta D, d)$-feasible demand matching, two adjacent edges cannot both be chosen since $B + \beta D = \lfloor (2 - \beta) D \rfloor - 1 + \beta D \leq (2 - \beta)D - 1 + \beta D = 2D - 1 < 2D$. Since $G$ is the triangle, any pair of edges is adjacent. Hence, every optimal $(b + \beta D, d)$-feasible demand matching $M^\star$ consists of exactly one edge, and thus has weight $1$. It follows that the ratio of $w(M^\star)$ to the optimum of \eqref{eq:lp} is at most
    \begin{multline*}
        \frac{w(M^\star)}{\sum_{e \in E} x_e^\star w(e)} = \frac{1}{3B/(2D)} = \frac{2D}{3B} = \frac{2D}{3(\lfloor (2 - \beta) D \rfloor - 1)} \leq \frac{2D}{3((2 - \beta) D - 2)} \\
        = \frac{2}{3(2 - \beta)} \cdot \frac{1}{1 - \frac{2}{D(2 - \beta)}} \leq \frac{2}{3(2 - \beta)} \left(1 + \frac{4}{D(2 - \beta)}\right) \leq \frac{2}{3(2 - \beta)} + \varepsilon,
    \end{multline*}
    where the second inequality follows from the fact that $1/(1 - x) \leq 1 + 2x$ for $x \in [0, 1/2]$ and from the assumption that $(2 - \beta) D \geq 4$. This completes the proof.
\end{proof}

\begin{lemma} \label{lem:tight-nba}
    Let $\beta \in [1, 4/3]$. Let $\varepsilon > 0$. Then there is an instance $(G, b, d, w)$ of the \textsc{Demand Matching} problem such that $w(M^\star)$ is at most $(6/(10 - 3\beta) + \varepsilon)$ times the optimum of \eqref{eq:lp}, where $M^\star$ is an optimal $(b + \beta d_{\max}, d)$-feasible demand matching. Moreover, this instance satisfies the no-bottleneck assumption.
\end{lemma}

\begin{proof}
    Let $G = (V, E)$ be a multigraph defined as follows. Let $V := \{ v_1, v_2, v_3, u \}$. Let $E_1 := \{ e_1, e_2, e_3 \}$, where $e_1 := v_1 v_2$, $e_2 := v_2 v_3$, and $e_3 := v_3 v_1$. Let $E_2 := \{ f_1, f_2 \}$, where $f_1 := v_1 v_2$, and $f_2 := v_3 u$. Let $E := E_1 \sqcup E_2$. Let $\varepsilon' := \min\{ 2/3, (10-3\beta)^2 \varepsilon/36 \}$. Let $\gamma \in \QQ \cap (\beta - 1, \beta - 1 + \varepsilon']$. Since $\beta \leq 4/3$ and $\varepsilon' \leq 2/3$, we have $\gamma \leq 1$. Let $D \in \NN$ be such that $\gamma D \in \NN$. Let $b(v_i) := D$ for all $i \in [3]$. Let $b(u) := \gamma D$. Let $d(e) := D$ and $w(e) := 1$ for all $e \in E_1$. Let $d(e) := \gamma D \leq D$ and $w(e) := 1$ for all $e \in E_2$. Then $d_{\max} = D$, and $(G, b, d, w)$ satisfies the no-bottleneck assumption. We illustrate this instance in \Cref{fig:tight-nba}.

    \begin{figure}[ht]
        \centering

        \begin{tikzpicture}[
            scale=1,
            vertex/.style={
                circle,
                draw,
                fill=white,
                inner sep=1.5pt,
                minimum size=22pt,
                font=\small
            },
            heavy/.style={thick},
            light/.style={thick, dashed}
        ]
        
        \node[vertex] (v1) at (0,0) {$D$};
        \node[vertex] (v2) at (4,0) {$D$};
        \node[vertex] (v3) at (2,3) {$D$};
        \node[vertex] (u)  at (4.4,3) {$\gamma D$};
        
        \draw[heavy] (v1) to[bend right=16] (v2);
        \draw[heavy] (v2) -- (v3);
        \draw[heavy] (v3) -- (v1);
        
        \draw[light] (v1) to[bend left=18] (v2);
        \draw[light] (v3) -- (u);
        
        \begin{scope}[shift={(5.5,1.65)}]
            \node[inner sep=6pt, align=left, font=\small] at (0,0) {
                \begin{tikzpicture}[baseline=-0.5ex]
                    \draw[heavy] (0,0) -- (0.75,0);
                    \node[right] at (0.85,0) {$(d,w)=(D,1)$};
                    \draw[light] (0,-0.45) -- (0.75,-0.45);
                    \node[right] at (0.85,-0.45) {$(d,w)=(\gamma D,1)$};
                \end{tikzpicture}
            };
        \end{scope}
        
        \end{tikzpicture}
        
        \caption{The instance in the proof of \Cref{lem:tight-nba}.}
        \label{fig:tight-nba}
    \end{figure}
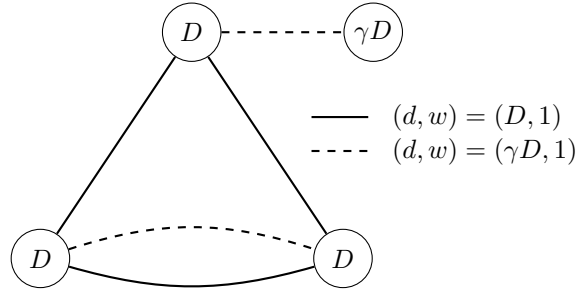

    Let $x^\star \in \RR^E$ be defined by $x_e^\star := (1 - \gamma)/2$ for all $e \in E_1$, and $x_e^\star := 1$ for all $e \in E_2$. Then
    \begin{align*}
        \sum_{e \in \delta(v_i)} x_e^\star d(e) &= 2 \cdot \frac{1 - \gamma}{2} \cdot D + 1 \cdot \gamma D = D = b(v_i), && \forall i \in [3], \\
        \sum_{e \in \delta(u)} x_e^\star d(e) &= \gamma D = b(u).
    \end{align*}
    Hence, $x^\star$ is feasible in \eqref{eq:lp}.
    
    For the weight of an optimal $(b + \beta d_{\max}, d)$-feasible demand matching, it is not hard to check that any subset $F \subseteq E$ with $|F| \geq 4$ is not a $(b + \beta d_{\max}, d)$-feasible demand matching. Hence, the weight of an optimal $(b + \beta d_{\max}, d)$-feasible demand matching $M^\star$ is at most $3 \cdot 1 = 3$. Therefore, the ratio of $w(M^\star)$ to the optimum of \eqref{eq:lp} is at most
    \begin{multline*}
        \frac{w(M^\star)}{\sum_{e \in E} x_e^\star w(e)} = \frac{3}{(7 - 3\gamma)/2} = \frac{6}{7 - 3\gamma} \leq \frac{6}{7 - 3(\beta - 1 + \varepsilon')} = \frac{6}{10 - 3\beta - 3\varepsilon'} \\
        = \frac{6}{10 - 3\beta} \cdot \frac{1}{1 - \frac{3\varepsilon'}{10 - 3\beta}} \leq \frac{6}{10 - 3\beta} \cdot \left(1 + \frac{6\varepsilon'}{10 - 3\beta}\right) = \frac{6}{10 - 3\beta} + \frac{36\varepsilon'}{(10 - 3\beta)^2} \leq \frac{6}{10 - 3\beta} + \varepsilon,
    \end{multline*}
    where the second inequality follows from the fact that $1 / (1 - x) \leq 1 + 2x$ for $x \in [0, 1/2]$ and from the assumptions that $\varepsilon' \leq 2/3$ and $\beta \leq 4/3$. This completes the proof.
\end{proof}

\Cref{lem:tight-general-2,lem:tight-nba} together prove \Cref{thm:tight-nba}.

\bibliographystyle{abbrvnat}
\bibliography{refs.bib}

\end{document}